\documentclass[a4paper,conference]{IEEEtran}

\usepackage[T1]{fontenc}
\usepackage[ansinew]{inputenc}
\usepackage{lmodern}
\usepackage{flushend}
\usepackage{graphicx}
\usepackage{tikz}
\usepackage{color}
\usepackage{amsmath, amsfonts, amssymb, mathrsfs}
\usepackage[amsmath,thmmarks]{ntheorem}
\usepackage{cite}

\theoremstyle{plain}
\theoremheaderfont{\itshape\bfseries}
\theorembodyfont{\normalfont\itshape}
\theoremseparator{:}
\theoremsymbol{}%\vspace{0.2cm}}
\newtheorem{theorem}{Theorem}
\newtheorem{lemma}[theorem]{Lemma}
\newtheorem{corollary}[theorem]{Corollary}

\theoremstyle{nonumberplain}

\newtheorem{condA}{Condition~A}
\newtheorem{condB}{Condition~B}

\theoremstyle{plain}
\theoremheaderfont{\itshape}
\theorembodyfont{\normalfont}
\theoremseparator{:}
\newtheorem{example}{Example}

\theoremstyle{nonumberplain}
\theoremheaderfont{\itshape}
\theorembodyfont{\normalfont}
\theoremseparator{:}
\theoremsymbol{}

\theoremstyle{nonumberplain}
\theoremsymbol{\rule{1ex}{1ex}}
\newtheorem{proofsketch}{Sketch of proof}
\newtheorem{proof}{Proof}
\def\Hx{ \widehat{x} }
\def\Hy{ \widehat{y} }
\def\Hbx{ \widehat{\mathbf{x}}}
\def\Hby{ \widehat{\mathbf{y}}}
\def\Hbu{ \widehat{\mathbf{u}}}
\def\d{ \mathrm{d} }							% integral d
\def\I{ \mathrm{i} }							% imaginary unit
\def\E{ \mathrm{e} }							% e
\def\T{ \mathrm{T} }							% T - transpose
\def\NN{ \mathbb{N} }						% positive integers
\def\ZN{ \mathbb{Z} }						% integers
\def\RN{ \mathbb{R} }						% real numbers
\def\CN{ \mathbb{C} }						% complex numbers
\def\TT{ \mathbb{T} }                               % subset of R
\def\L{ \mathcal{L} }								% Hilbert space 3
\def\ourPW{ \mathcal{PW}_{T/2} }
\def\ourL2{ \mathcal{L}^{2}(\TT)}

\def\ourRT{ \mathcal{R}_{M}(\TT)}
\def\ourB{ \mathcal{B}^{\infty}_{T'/2}}
\def\bsc{ \boldsymbol{c} }						% Sequence c
\def\balpha{ \boldsymbol{\alpha} }
\def\blambda{ \boldsymbol{\lambda} }
\def\bx{ \mathbf{x} }								% bold x
\def\bc{ \mathbf{c} }								% bold c
\DeclareMathOperator*{\sinc}{sinc}							% sinc function
\begin{document}

\title{Phase Retrieval via Structured Modulations\\ in Paley-Wiener Spaces}

\author{\IEEEauthorblockN{Fanny~Yang, Volker~Pohl, Holger~Boche}
\IEEEauthorblockA{Lehrstuhl f{\"u}r Theoretische Informationstechnik\\
	       Technische Universit{\"a}t M{\"u}nchen, 80290 M{\"u}nchen, Germany\\
	       { \{fanny.yang, volker.pohl, boche\}@tum.de} }}

\maketitle
% ----- Abstract -----
\begin{abstract}
This paper considers the recovery of continuous time signals from the magnitude of its samples.
It uses a combination of structured modulation and oversampling and provides sufficient conditions on the signal and the sampling system such that signal recovery is possible.
In particular, it is shown that an average sampling rate of four times the Nyquist rate is sufficient to reconstruct a signal from its magnitude measurements.
\end{abstract}

% ================================================================================================================
% ========== INDRODUCTION ========================================================================================
% ================================================================================================================
\section{Introduction}
\label{sec:Intro}

In many applications, only intensity measurements are available to reconstruct a desired signal $x$. 
This is widely known as the phase retrieval problem which e.g. occurs in diffraction imaging applications such as X-ray crystallography, astronomical imaging, or in speech processing. 

In the past, concrete efforts have been made on the recovery of finite dimensional signals from the modulus of their Fourier transform. 
In general however, they require strong limitations on the signal such as constraints on its z-transforms \cite{Oppenheim_Phase_80} or knowledge of its support \cite{Fienup1982_PhaseRetrieval}. 
Furthermore, analytic frame-theoretic approaches were considered in \cite{Balan1,Balan_Painless} which require that the number of measurements grows proportionally with the square of the space dimension. Ideas of sparse signal representation and convex optimization where applied in \cite{Vetterli2011_PRbeyondFienup,CandesEldar_PhaseRetrieval} to allow for lower computational complexity.

Note that all the above approaches addressed finite dimensional signals and the question is whether similar results can be obtained for continuous signals in infinite dimensional spaces.
In \cite{Thakur2011} it was shown that real valued bandlimited signals are completely determined by its magnitude samples taken at twice the Nyquist rate.
For complex valued signals, the results for finite dimensional spaces \cite{Balan1,Balan_Painless,CandesEldar_PhaseRetrieval} indicate that oversampling alone may not be sufficient for signal reconstruction from magnitude samples, since their particular choice of measurement vectors was the key to enabling signal recovery.

In this work we are looking at continuous signals in Paley-Wiener spaces.
Our approach extends ideas from \cite{Balan1,Balan_Painless,CandesEldar_PhaseRetrieval} by applying a bank of modulators before the intensity measurements.
Then we are able to reconstruct the signals from samples taken at a rate of four times the Nyquist rate.

Basic notations for sampling and reconstruction in Paley-Wiener spaces are recaptured in Sec.~\ref{sec:PWSpaces}, Sec.~\ref{sec:Measurement} describes our sampling setup.
In Sec.~\ref{sec:MainResult} we provide sufficient conditions for perfect signal reconstruction from magnitude measurements of the Fourier Transform.
The paper closes with a short discussion in Sec.\ref{sec:Summary}.

% ================================================================================================================
% ========== Paley-Wiener ========================================================================================
% ================================================================================================================
\section{Sampling in Paley-Wiener Spaces}
\label{sec:PWSpaces}

Let $\mathbb{S} \subseteq \RN$ be an arbitrary subset of the real axis $\RN$.
For $1 \leq p \leq \infty$ we write $\L^{p}(\mathbb{S})$ for the usual Lebesgue space on $\mathbb{S}$.
In particular, $\L^{2}(\mathbb{S})$ is the Hilbert space of square integrable functions on $\mathbb{S}$ with the inner product
\begin{eqnarray*}
%\label{equ:InnerProd}
	& \left\langle x,y \right\rangle_{\L^{2}(\mathbb{S})} = \int_{\mathbb{S}} x(\theta)\, \overline{y(\theta)}\, \d\theta
\end{eqnarray*}
where the bar denotes the complex conjugate. In finite dimensional spaces $\langle x,y \rangle = y^*x$ where $^*$ denotes the conjugate transpose.
Let $T>0$ be a real number. Throughout this paper $\TT = [-T/2,T/2]$ stands for the closed interval of length $T$, and
$\ourPW$ denotes the \emph{Paley-Wiener space} of entire functions of exponential type $T/2$ whose restriction to $\RN$ belongs to $\L^2(\RN)$.
The Paley-Wiener theorem %(see, e.g., \cite{Levin1997_Lectures,Young2001_Nonharmonic}) 
states that to every $\Hx \in \ourPW$ there exists an $x \in \ourL2$ such that
\begin{eqnarray}
\label{equ:PaleyWiener}
& \Hx(z) = \int_{\TT}x(t)\,\E^{\I tz}\,\d t\,,
\quad\text{for all}\ z\in\CN\;.
\end{eqnarray}
In the following we will call $x$ the signal in the \emph{time domain} and $\Hx$ the signal in the \emph{Fourier domain}, since its restriction to the real axis is a Fourier Transform.

If not otherwise noted, our signal space will be $\ourL2$, i.e. we consider signals of finite energy which are supported on the finite interval $\TT$. 
This is a natural assumption since signals in reality are usually finite in time.
It follows from the Paley-Wiener theorem that to every signal $x \in \ourL2$ there corresponds an $\Hx \in \ourPW$ given by \eqref{equ:PaleyWiener}.

A sequence $\Lambda = \{\lambda_{n}\}_{n\in\ZN}$ of complex numbers is said to be \emph{complete interpolating} for $\ourPW$ if and only if the functions $\{\phi_n(t) := \E^{-\I\lambda_n t}\}_{n\in \ZN}$ form a Riesz basis for $\ourL2$ \cite{Young2001_Nonharmonic}. Let $x \in \ourL2$ be arbitrary. Then it follows from \eqref{equ:PaleyWiener} that
\begin{eqnarray*}
& \Hx(\lambda_{n})
%= \int_{\TT} x(t)\, \E^{\I t \lambda_{n}}\, \d t
= \left\langle x , \phi_{n} \right\rangle_{\ourL2}
\quad\text{for all}\ n\in\ZN\;.
\end{eqnarray*}
Since $\{\phi_{n}\}_{n\in\ZN}$ is a Riesz basis for $\ourL2$ the signal $x$ can be reconstructed from the samples $\Hx(\Lambda) = \{\Hx(\lambda_{n})\}_{n\in \ZN}$ by
\begin{eqnarray}
\label{equ:TimeReconstruct}
	& x(t)
	= \sum_{n\in\ZN} \left\langle x,\phi_{n}\right\rangle\, \psi_{n}(t)
	= \sum_{n\in\ZN} \Hx(\lambda_{n})\, \psi_{n}(t)\;,
\end{eqnarray}
where $\{\psi_{n}\}_{n\in\ZN}$ is the unique dual Riesz basis of $\{\phi_{n}\}_{n\in\ZN}$ \cite{Christensen_Frames}. It is well-known that in the Fourier domain 
\begin{equation*}
	\widehat{\psi}_{n}(z)
	= \frac{S(z)}{S'(\lambda_{n}) (z - \lambda_{n})}
	\ \text{with}\
	S(z) = z^{\delta_\Lambda} \!\!\! \lim_{R\to\infty} \!\! \prod_{\substack{|\lambda_{n}| < R \\ \lambda_{n} \neq 0}} \!\!\!\! \big(1 - \frac{z}{\lambda_{n}} \big)
\end{equation*}
with $\delta_\Lambda = 1$ if $0\in \Lambda$ and $\delta_{\Lambda} = 0$ otherwise, where the \emph{generating function} $S$ is an entire function of exponential type $T/2$.
The infinite product converges uniformly on compact subsets of $\CN$ if $\Lambda$ is a complete interpolating sequence (see \cite{Pavlov_Convergence}). %\cite{Levin1997_Lectures}).
Note that in order to obtain the signal in the time domain by \eqref{equ:TimeReconstruct}, it is sufficient to find the inverse Fourier transform of $\widehat{\psi}_n(\omega), \omega \in \RN$.

\begin{example}
\label{exa:Shannon}
The well known Shannon sampling series is obtained for regular sampling with $\lambda_{n} = n \tfrac{2\pi}{T}$, $n\in\ZN$.
Then $S(z) = \sin(\frac{T}{2} z)$ and $\widehat{\psi}_n(z) = \sinc(\frac{T}{2}[z-n\frac{2\pi}{T}])$ where $\sinc(x) := \sin(x)/x$.
This corresponds to modulated rectangular functions in the time domain such that \eqref{equ:TimeReconstruct} becomes
$x(t) = \sum_{n\in\ZN} \Hx(\lambda_{n})\, \E^{-\I n \frac{2\pi}{T} t}$ for all $t\in\TT$.
\end{example}

% ================================================================================================================
% ========== Measurement =========================================================================================
% ================================================================================================================
\section{Measurement Methodology}
\label{sec:Measurement}

We apply a measurement methodology which uses oversampling in connection with structured modulations of the desired signal, inspired by the approach in \cite{CandesEldar_PhaseRetrieval}.
Suppose $x \in \ourL2$ is the signal of interest. Although the loss of phase information is intrinsic to the measurement procedure, it is possible to influence the desired signal before the actual measurement.
More precisely in our sampling scheme in Fig.~\ref{fig:MeasureSetup}, we assume that the signal of interest $x$ is multiplied with $M$ known modulating functions $p^{(m)}$. In optics, these modulations may be different diffraction gratings between the object (the desired signal) and the measurement device \cite{CandesEldar_PhaseRetrieval}.
This way we obtain a collection of $M$ representations (or illuminations) $y^{(m)}$ of $x$. Afterwards, the modulus of the Fourier spectra $\widehat{y}^{(m)}$ are measured and uniformly sampled with frequency spacing $\beta$.

%
% ----- FIGURE - Sampling Scheme -----
\begin{figure}[t]
\begin{center}
\begin{picture}(240,95)(10,25)
	% ----- Begin -----
	\put(20,77){\makebox(0,0){\footnotesize $x(t)$}}
	\put(10,70){\vector(1,0){20}}
	\put(30,40){\line(0,1){60}}
	% ===== Modulatoren ======
	% --- Oben ---
	\put(30,100){\vector(1,0){21}}
	\put(55,100){\circle{8}}		\put(55,100){\makebox(0,0){\footnotesize$\times$}}
	\put(55,85){\vector(0,1){11}}	\put(65,80){\makebox(0,0){\footnotesize $p^{(1)}(t)$}}
	\put(59,100){\line(1,0){36}}
	\put(80,107){\makebox(0,0){\footnotesize $y^{(1)}$}}
	% --- dots ---
	\put(55,70){\makebox(0,0){\footnotesize $\vdots$}}
	% --- Unten ---
	\put(30,40){\vector(1,0){21}}
	\put(55,40){\circle{8}}   	\put(55,40){\makebox(0,0){\footnotesize$\times$}}
	\put(55,25){\vector(0,1){11}}	\put(65,20){\makebox(0,0){\footnotesize $p^{(M)}(t)$}}
	\put(59,40){\line(1,0){36}}
	\put(80,47){\makebox(0,0){\footnotesize $y^{(M)}$}}
	% ===== Intensity Measurement =====
	% --- Oben ---
	\put(95,110){\line(1,0){20}}
	\put(95,90){\line(1,0){20}}
	\put(95,90){\line(0,1){20}}
	\put(115,90){\line(0,1){20}}
	\put(105,100){\makebox(0,0){\footnotesize IM}}
	\put(115,100){\line(1,0){40}}
	\put(135,109){\makebox(0,0){\footnotesize $\big|\widehat{y}^{(1)}\big|^{2}$}}
	% --- dots ---
	\put(105,70){\makebox(0,0){\footnotesize $\vdots$}}	
	% --- Unten ---
	\put(95,50){\line(1,0){20}}
	\put(95,30){\line(1,0){20}}
	\put(95,30){\line(0,1){20}}
	\put(115,30){\line(0,1){20}}
	\put(105,40){\makebox(0,0){\footnotesize IM}}
	\put(115,40){\line(1,0){40}}
	\put(135,49){\makebox(0,0){\footnotesize $\big|\widehat{y}^{(M)}\big|^{2}$}}
	% ===== Sampling =====
	% --- Oben ---
	\put(155,100){\line(2,1){10}}
	\qbezier(156,106)(161,105)(162,100)	\put(162,98){\vector(0,-1){2}}
	\put(165,100){\vector(1,0){20}}
	\put(205,107){\makebox(0,0){\footnotesize $c^{(1)}_{n} = |\widehat{y}^{(1)}(n\beta)|^{2}$}}
	% --- dots ---
	\put(160,70){\makebox(0,0){\footnotesize $\vdots$}}		
	% --- Unten ---
	\put(155,40){\line(2,1){10}}
	\qbezier(156,46)(161,45)(162,40)		\put(162,38){\vector(0,-1){2}}
	\put(165,40){\vector(1,0){20}}
	\put(210,47){\makebox(0,0){\footnotesize $c^{(M)}_{n}  = |\widehat{y}^{(M)}(n\beta)|^{2}$}}
\end{picture}
\end{center}
\caption{Measurement setup: In each branch, the unknown signal $x$ is modulated with a different sequence $p^{(m)}$, $m=1,2,\dots,M$.
Subsequently, the intensities of the resulting signals $y^{(m)}$ are measured (IM) and uniformly sampled in the frequency domain.}
\label{fig:MeasureSetup}
\end{figure}
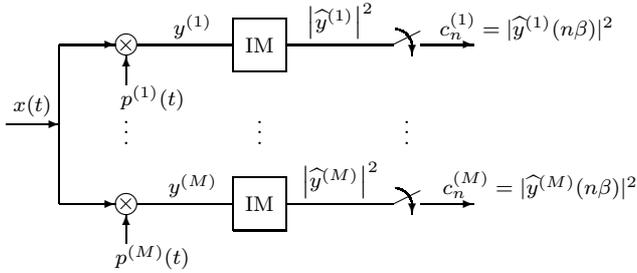

Let $p^{(m)}$ have the following general form
\begin{eqnarray}
\label{equ:ModSig}
	& p^{(m)}(t) := \sum^{K}_{k=1} \overline{\alpha^{(m)}_{k}} \E^{\I\lambda_k t}
\end{eqnarray}
where $\lambda_{k}$ and $\alpha^{(m)}_{k}$ are complex coefficients which are determined subsequently. 
Consequently the samples in the $m$th branch are given by
\begin{eqnarray}
\label{equ:IntensMeasure}
	c^{(m)}_{n}
	&=& |\widehat{y}^{(m)}(n \beta)|^{2}
	= \left| \sum^{K}_{k=1} \overline{\alpha^{(m)}_{k}}\,  \Hx(n \beta + \lambda_{k}) \right|^{2}\nonumber\\
	&=& |\langle  \widehat{\bx}_{n},\balpha^{(m)} \rangle|^2
\end{eqnarray}
with the length $K$ vectors
\begin{equation*}
\balpha^{(m)} := \left(\begin{array}{c}
	\alpha_1^{(m)} \\ \vdots \\ \alpha_K^{(m)}
	\end{array}\right)
\ \ \text{and}\ \
\Hbx_{n} := \left(\begin{array}{c}
	\Hx(n\beta + \lambda_1) \\ \vdots \\ \Hx(n\beta + \lambda_K)
	\end{array}\right)\;.
\end{equation*}
We will show that if $\balpha^{(m)}$ and the interpolation points $\lambda_{n,k} := n\beta + \lambda_{k} : n\in\ZN\;, k=1,\dots,K$ are properly chosen, it is possible to reconstruct $x$ from all samples $\bc = \{ c^{(m)}_{n} : m=1,\dots,M\, ;\, n\in\ZN \}$. 
%% The reconstruction procedure then consists of two steps.
%% First for each $n\in \ZN$ one determines the vector $\widehat{\bx}_{n}$ from the $M$ measurements $\bc^{(m)} := \{c^{(m)}_{n}\}_{m=1}^M$  \eqref{equ:IntensMeasure}.
%% Secondly, we reconstruct the continuous $x$ from the entries of the vectors $\widehat{\bx}_{n}$. This is only possible if the sampling points are chosen appropriately.

% ---------------------------------------------------------------------------------
\subsection{Choice of the coefficients $\alpha^{(m)}_{k}$} 
\label{sec:Balan}
In order to determine the vector $\Hbx_{n} \in \CN^{K}$ from the $M$ intensity measurements $\bc^{(m)}$, we apply a result from \cite{Balan_Painless}.
It states that if the family of $\CN^{K}$-vectors $\mathcal{A} = \{\balpha^{(1)}, \dots,\balpha^{(M)} \}$ 
%gives rise to a \emph{projective $2$-design}, then every vector $\Hbx_{n} \in \CN^{K}$ can be reconstructed up to a constant phase from the magnitude of the inner products \eqref{equ:IntensMeasure}. 
%In \cite{Balan_Painless}, two types of frames are shown to fulfill this condition. In one case $\mathcal{A}$ 
constitutes a $2$-uniform $M/K$-tight frame which contains $M = K^{2}$ vectors or $\mathcal{A}$ is a union of $K+1$ mutually unbiased bases in $\CN^{K}$, then every $\Hbx_{n} \in \CN^{K}$ can be reconstructed up to a constant phase from the magnitude of the inner products \eqref{equ:IntensMeasure}.
For simplicity, we will only discuss the first case here and therefore fix $M = K^{2}$.
\begin{condA}
A sampling system as in Fig.~\ref{fig:MeasureSetup} is said to satisfy Condition~A if the coefficients $\alpha^{(m)}_{k}$ in \eqref{equ:ModSig} are such that $\mathcal{A}$ constitutes a $2$-uniform $M/K$-tight frame.
\end{condA}
Then reconstruction will be based on the following formula
\begin{equation}
\label{equ:BalanReconstruct}
Q_{\Hat{\bx}_{n}} %&= \frac{(K+1)}{K}\sum_{m=1}^{M}c_n^{(m)}Q_{\balpha^{(m)}} - \|\Hat{\bx}(n\beta)\|^2I\\
= \frac{(K+1)}{K}\sum_{m=1}^{M}c_n^{(m)}Q_{\balpha^{(m)}} - \frac{1}{K}\sum_{m=1}^{M}c_n^{(m)}I
\end{equation}
with rank-$1$ matrices $Q_\bx = \bx\bx^{*}$. 
For $K=2$ a valid choice  for $\mathcal{A}$ reads \cite{Balan_Painless}
\begin{equation*}
	\balpha^{(1)} =  \binom{\alpha}{\beta},\
	\balpha^{(2)} =  \binom{\beta}{\alpha},\
	\balpha^{(3)} =  \binom{\alpha}{-\beta},\
	\balpha^{(4)} =  \binom{-\beta}{\alpha}
\end{equation*}
with $\alpha = \sqrt{\frac{1}{2}(1-\frac{1}{\sqrt{3}})}$ and $\beta = \E^{\I 5\pi/4}\sqrt{\frac{1}{2}(1+\frac{1}{\sqrt{3}})}$.

% ---------------------------------------------------------
\subsection{Choice of the interpolation points}
\label{sec:Sampling}

Let $\{\lambda_k\}_{k=1}^K$ be ordered increasingly by their real parts.
%us write $\lambda_k = \xi_k + \I \eta_k$ and without loss of generality order the set $\{\lambda_k\}_{k=1}^K$ according to their real parts such that $\cdots \leq \xi_{k-1} \leq \xi_{k} \leq \cdots$.
For each $n\in\ZN$, the vector $\Hbx_{n}$ contains the values of $\Hx$ at $K$ distinct interpolation points in the complex plane
%With each frequency sampling at $n\beta$, we obtain the Fourier modulus at $K$ sampling points,
collected in the sequences
\begin{equation}
\label{equ:SampSets}
	\blambda^a_n := \{\lambda^a_{n,k}\}^{K}_{k = 1}
	\quad\text{with}\quad
	\lambda^a_{n,k} =  n\beta + \lambda_k\;,
	\quad n\in\ZN\;.
\end{equation}
Therein, the parameter $a\in\NN$ denotes the number of overlapping points of consecutive sets \eqref{equ:SampSets} (cf. also Fig.\ref{fig:SpampPoints}).
More precisely, we require for every $n\in\ZN$ that
\begin{equation}
\label{equ:CondSampPoint}
\lambda^a_{n,i} = \lambda^a_{n-1,K-i+1}\quad \text{for all}\ i = 1,\dots, a\;.
\end{equation}
In the following we denote by $\Lambda_{O,n}^a = \blambda^a_n \cap \blambda^a_{n+1}$ the set of overlapping interpolation points between $\blambda^a_n$ and $\blambda^a_{n+1}$,
and we define the overall interpolation sequences
\begin{equation*}
\Lambda^a := \cup_{n\in\ZN}\blambda^a_n\;.
\end{equation*}
In general we allow for $a\geq 1$, but we will see that $a=1$ is generally sufficient for reconstruction.

As explained in Sec.~\ref{sec:PWSpaces}, $x\in\ourL2$ can be perfectly reconstructed by \eqref{equ:TimeReconstruct} if $\Lambda^a$ is complete interpolating for $\ourPW$.

% ----- Condition SamplingPoints -----
\begin{condB}
A sampling system as in Fig.~\ref{fig:MeasureSetup} is said to satisfy Condition~B if the coefficients $\{\lambda_{k}\}^{K}_{k=1}$ in \eqref{equ:ModSig} are such that $\Lambda^{a}$ is complete interpolating for $\ourPW$ and satisfies \eqref{equ:CondSampPoint} for a certain $1 \leq a < K$.
\end{condB}
% ------------------------------------

Note that interpolating sequences fulfilling Condition B are $\beta$-periodic.
In general it is hard to characterize sets which fulfill this condition. However one specific example is when $\Lambda^a$ equals to the set of zeros of a $\beta$-periodic sine-type function of type $T'/2 \leq T/2$ (see, e.g., \cite{Young2001_Nonharmonic,Levin1997_Lectures}). Sine-type functions are entire functions $f$ of exponential type $T'/2$ with simple and isolated zeros and 
for which there exist positive constants $A,B,H$ such that
\begin{equation*}
%\label{equ:SineTypeH}
Ae^{\frac{T'}{2}|\eta|} \leq |f(\xi + \I \eta)| \leq B e^{\frac{T'}{2}|\eta|}\;,
\quad\text{for}\ |\eta| > H\;.
\end{equation*}
Moreover, sequences which differ from interpolation patterns derived from sine-type functions only by their imaginary parts, are still zeros of sine-type functions \cite{Levin_Perturbations}. The complete interpolating property is also preserved under little shifts in the real part using Katsnelson's theorem \cite{Levin1997_Lectures}.
Apparently it is possible to construct many non-uniform complex interpolation sequences $\Lambda^{a}$ which satisfy Condition~B. One particular simple construction is obtained by starting with the zeros of the sine-type function $\sin( \frac{T'}{2}z )$, which has equally spaced zeros on the real axis (cf. Example~\ref{exa:Shannon}).

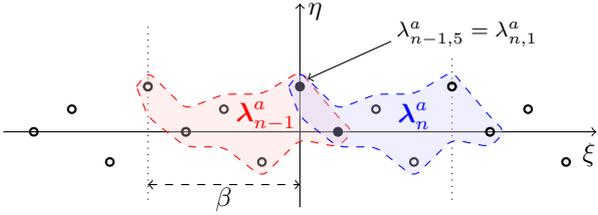
\begin{figure}[t]
\begin{center}
\begin{tikzpicture}
  % --- LEFT ---
	% ----- Re & Im Achse -----
	\draw[->] (-3.9,0) -- (3.9,0);
	\draw (3.8,-0.3) node {$\xi$};
	\draw[->] (0,-1.0) -- (0,1.7);
	\draw (0.2,1.6) node {$\eta$};
	\draw[dotted] (-2,-0.9) -- (-2,1.4);
	\draw[dotted] (2,-0.9) -- (2,1.0);
	\draw[<->,dashed] (-2,-0.7) -- (0,-0.7);
	\draw (-1,-0.9) node {$\beta$};
	% ----- Sets -----
	% ----- 1 Set of Sampling Points -----
	\draw[black, thick] (-3.5,0.0) circle (0.05);
	\draw[black, thick] (-3.0,0.3) circle (0.05);
	\draw[black, thick] (-2.5,-0.4) circle (0.05);
	\draw[black, thick] (-2.0,0.6) circle (0.05);
	\draw[black, thick] (-1.5,0.0) circle (0.05);		
	% ----- (n-1)th Set of Sampling Points -----
	\draw[black, thick] (-2.0,0.6) circle (0.05);
	\draw[black, thick] (-1.5,0.0) circle (0.05);
	\draw[black, thick] (-1.0,0.3) circle (0.05);
	\draw[black, thick] (-0.5,-0.4) circle (0.05);
	\draw[black, thick] (0.0,0.6) circle (0.05);
	\draw[black, thick] (0.5,0.0) circle (0.05);
	\filldraw[fill = red!20!white, fill opacity=0.3, draw = red, dashed, rounded corners = 1mm]
	           (-2.15,0.6) -- (-2.0,0.8) -- (-1.5,0.3) -- (-1.0,0.5) -- (-0.5,0.5) -- (0.0,0.8) -- (0.5,0.2) -- (0.7,0.0) --
	           (0.5,-0.2) -- (0.0,-0.1) -- (-0.5,-0.6) -- (-1.0,-0.2) -- (-1.5,-0.2) -- (-2.0,0.2) -- (-2.15,0.6);
	\draw[red] (-0.45,0.21) node{$\blambda^{a}_{n-1}$};
	% ----- nth Set of Sampling Points -----
	\filldraw[black, thick] (0.0,0.6) circle (0.05);
	\filldraw[black, thick] (0.5,0.0) circle (0.05);
	\draw[black, thick] (1.0,0.3) circle (0.05);
	\draw[black, thick] (1.5,-0.4) circle (0.05);
	\draw[black, thick] (2.0,0.6) circle (0.05);
	\draw[black, thick] (2.5,0.0) circle (0.05);
	\filldraw[fill = blue!20!white, fill opacity=0.3, draw = blue, dashed, rounded corners = 1mm]
	           (-0.15,0.6) -- (0.0,0.8) -- (0.5,0.3) -- (1.0,0.5) -- (1.5,0.5) -- (2.0,0.8) -- (2.5,0.2) -- (2.7,0.0) --
	           (2.5,-0.2) -- (2.0,-0.1) -- (1.5,-0.6) -- (1.0,-0.2) -- (0.5,-0.2) -- (0.0,0.2) -- (-0.15,0.6);
	\draw[blue] (1.5,0.21) node{$\blambda^{a}_{n}$};
	% ---
	\draw[black, thick] (2.0,0.6) circle (0.05);
	\draw[black, thick] (2.5,0.0) circle (0.05);
	\draw[black, thick] (3.0,0.3) circle (0.05);
	\draw[black, thick] (3.5,-0.4) circle (0.05);
	% -----
	%\draw[black] (-2.0,1.0) node{\footnotesize $\lambda^{a}_{n-1,1}$};
	%\draw[black] (0.0,1.0) node{\footnotesize $\lambda^{a}_{n-1,5} = \lambda^{a}_{n,1}$};
	\draw[->] (1.2,1.2) -- (0.1,0.7);
	\draw[black] (2.2,1.3) node{\footnotesize $\lambda^{a}_{n-1,5} = \lambda^{a}_{n,1}$};
\end{tikzpicture}
\end{center}
\caption{Illustration for the choice of interpolation points in the complex plane for $K=6$ in \eqref{equ:ModSig} and an overlap $a = 2$.}
\label{fig:SpampPoints}
\end{figure}

% ================================================================================================================
% ========== MAIN RESULT =========================================================================================
% ================================================================================================================
\section{Signal Reconstruction -- Main Results}
\label{sec:MainResult}

We assume a sampling scheme as described in Section~\ref{sec:Measurement} (cf.~Fig.~\ref{fig:MeasureSetup}) which satisfies Condition~A and B.
For this setup, we show that basically every $x \in \ourL2$ can be reconstructed from the samples in \eqref{equ:IntensMeasure}.
The proof provides an explicit algorithm for perfect signal recovery. 

% ---------- Main Theorem ----------
\begin{theorem}
\label{thm:MainThm}
Let $x \in \ourL2$ be sampled according to the scheme in Section~\ref{sec:Measurement} which satisfies Condition~A and B, and
let $\bc = \{c^{(m)}_n\}^{m=1,\dots,M}_{n\in\ZN}$ be the sampling sequence in \eqref{equ:IntensMeasure}. 
If the set $\Hx(\Lambda^a_{O,n})$ contains at least one non-zero element for each $n\in\ZN$, then $x$ can be perfectly reconstructed from $\bc$ up to a constant phase.
\end{theorem}
% ----------------------------------

\begin{proof}
According to Condition~B of the sampling system, $\Lambda^{a}$ is complete interpolating for $\ourPW$. Therefore the signal $x$ can be reconstructed from the vectors $\{\Hbx_{n}\}_{n\in\ZN}$ using
\eqref{equ:TimeReconstruct}. It remains to show that $\{\Hbx_{n}\}_{n\in\ZN}$ can be determined from $\bsc$.

Let $n\in\ZN$ be arbitrary. Since the sampling system satisfies Condition~A, we can use \eqref{equ:BalanReconstruct} to obtain the rank-$1$ matrix $Q_{n} := \Hbx_n\Hbx_n^*$ from the measurements $\{c^{(m)}_{n}\}^{M}_{m=1}$. Then $\Hbx_n \in \CN^K$ is obtained by factorizing $Q_{n}$. However, such a factorization is only unique up to a constant phase factor.
If the phase $\phi_{n,i}$ of one element $[\Hbx_n]_i$ is known, the vector $\Hbx_n$ can be completely determined from $Q_{n}$ by
\begin{equation}
\label{equ:factorisation}
\Hx(n\beta +\lambda_k) = \sqrt{[Q_{n}]_{k,k}}\, \E^{\I(\phi_{n,i} - \arg([Q_{n}]_{i,k}))}, \ \forall k \neq i\,.
\end{equation}

Assume that we start the determination of the sequence $\{\Hbx_{n}\}_{n\in\ZN}$ at a certain $n_{0}\in\ZN$.
In this initial step, we set the constant phase of $\Hbx_{n_{0}}$ arbitrarily.
In the next step, we determine $\Hbx_{n_{0}+1}$. After the factorization of $Q_{n_{0}+1}$, the vector $\Hbx_{n_{0}+1}$ is only determined up to a constant phase. 
However, since $\Lambda^{a}_{O,n_0}$ is non-empty, and because $\Hx(\Lambda^a_{O,n_{0}})$ contains at least one non-zero element, we have phase knowledge of at least one entry of $\Hbx_{n_{0}+1}$, say $\Hx(\lambda^a_{n_{0}+1,i})$, where $\lambda^a_{n_{0}+1,i}$ is an overlapping interpolation point of $\blambda^a_{n_{0}}$ and $\blambda^a_{n_{0}+1}$.
Thus, we can completely determine $\Hbx_{n_{0}+1}$ and successively all  $n = n_{0}\pm 1, n_{0}\pm 2, \dots$ using \eqref{equ:factorisation} to obtain $\Hx(\Lambda^a)e^{\I \theta_0}$.
%In this way, we successively determine all $\Hbx_{n}$ for $n = n_{0}+1, n_{0}+2, \dots$, and in a similar way for $n = n_{0} - 1, n_{0}-2, \dots$.

Note that the arbitrary setting of the phase of the initial vector $\Hbx_{n_0}$ yields a constant phase shift $\theta_0 \in [-\pi,\pi]$ for all $\Hbx_{n}$ which persists also after the reconstruction of the time signal by \eqref{equ:TimeReconstruct}.
\end{proof}

Theorem~\ref{thm:MainThm} states that $x\in\ourL2$ can only be reconstructed if $\Hx\in\ourPW$ has at most $a-1$ zeros on the overlapping interpolation sets $\Lambda^{a}_{O,n}$.
However, this restriction is not too limiting, since it is known that the zeros of an entire function of exponential type can not be arbitrarily dense in $\CN$.
For example, defining $\mathcal{Z}_{n} := \{z \in \CN : n \pi/T < |z| \leq (n+1)\pi/T\}$, the result in \cite{Supper2002_Zeros} states that for every $x \in \ourPW$ there exist only finitely many sets $\mathcal{Z}_{n}$ which contain more than one zero of $x$.
Consequently, by choosing the spacing of the interpolation points in the overlapping sets $\Lambda^{a}_{O,n}$ less than $\pi/T$,  it is very unlikely that a randomly chosen function from $\ourPW$ fails to satisfy the condition of Theorem~\ref{thm:MainThm}, especially for $a>1$.
%For example, the result in \cite{Supper2002_Zeros} states that if $\mathcal{Z}_{n} := \{z \in \CN : n \pi/T < |z| \leq (n+1)\pi/T\}$ then for every $x \in \ourPW$ there exists infinitely many sets $\mathcal{Z}_{n}$ which contain less than one zero of $x$.
%Consequently, by choosing the spacing of the interpolation points in the overlapping sets $\Lambda^{a}_{O,n}$ less than $\pi/T$,  it will be very unlikely that a randomly chosen function from $\ourPW$ fails to satisfy the condition of Theorem~\ref{thm:MainThm}, especially for $a>1$.
%Surely, it can happen that in finitely many sets $\mathcal{Z}_{n}$, the zeros of $\Hx$ exactly sit on the overlapping interpolation points. But generically, this will not the case.

To avoid such pathological cases, we may a priori restrict the function space allowed in Theorem~\ref{thm:MainThm} to prevent $x$ from having zeros in $\Lambda^a$. 
To this end, we first state an extension of a well-known lemma by Duffin, Schaeffer \cite{DuffinSchaeffer_38}. The straightforward proof is omitted here.

% ---------- LEMMA - Duffin-Schaeffer ----------
\begin{lemma}
\label{lem:DufSchaef}
 Let $\Hx(z) \in \ourPW$ be an entire function of $z=\xi+\I \eta$ satisfying $|\Hx(\xi)|\leq M$ on the real axis. 
Then for every $T'>T$ the function
\begin{equation}
\label{equ:PlusCos}
  \Hy(z) = M\cos(\tfrac{T'}{2}z) - \Hx(z)
\end{equation}
belongs to the Bernstein space $\ourB$ and there exists a constant $H = H(T,T')$ such that $|\Hy(z)|>0\; \forall z:|\eta|>H$.
\end{lemma}
% ----------------------------------------------
The Bernstein space $\ourB$ is the set of all entire functions of exponential type $T'/2$ whose restriction to $\RN$ is in $\L^{\infty}(\RN)$.
Upon this we can establish the following corollary for functions in $\ourRT = \{x\in \ourL2:\|x\|_{\L^1(\TT)}\leq M\}$.
% ---------- Corollary 2 ----------
\begin{corollary}
\label{cor:Cor2}
Let $x\in \ourRT$ be sampled according to the scheme in Sec.~\ref{sec:Measurement}.
Then there exist interpolation sequences $\Lambda^{a}$ with overlap $a\geq1$ such that every $x \in \ourRT$ can be perfectly reconstructed (up to a constant phase) from the measurements \eqref{equ:IntensMeasure}.
\end{corollary}

\begin{proofsketch}
The bounded $\L^{1}$-norm of $x$ implies that $|\Hx(\xi)| \leq M$ for all $\xi\in\RN$.
Fix $T'>T$, then $\Hy$, defined by \eqref{equ:PlusCos}, has zeros only for $|\eta| \leq H$ by Lemma~\ref{lem:DufSchaef}.
Choose $\Lambda^a$ as the zero set of a sine-type function of type $T'/2$. 
Then by \cite{Levin_Perturbations} we can shift the imaginary parts of the interpolation points such that $|\eta_k|>H$ for all $k$ while $\Lambda^a$ remains to be the zero set of a sine-type function.
Now assume that the measurement is preceded by an addition of a cosine as in \eqref{equ:PlusCos} such that as in the proof of Theorem~\ref{thm:MainThm}, the factorization step yields 
$\Hby_{n} := [\Hy(n\beta + \lambda_{1}),\dots,\Hy(n\beta + \lambda_{K})]^{\T} = M\Hbu_n - \Hbx_n$ with $\Hbu_{n} := [\cos(\frac{T}{2}[n\beta + \lambda_{1}]),\dots,\cos(\frac{T}{2}[n\beta + \lambda_{K}])]^{\T}$.
We then obtain $\Hy(\Lambda^a)e^{\I \theta_0}$ since the overlapping interpolation points do not coincide with zeros of $\Hx$ and the phase information can be propagated. 
Since $\Hy\in \ourB$ and because $\Lambda^a$ is the zero set of a sine-type function, we know by \cite[Theorem~2]{MoenichBoche_SP2010} that 
it is possible to stably reconstruct $\Hy(z)\,\E^{\I \theta_0}$ from $\Hy(\Lambda^a)\,\E^{\I \theta_0}$ using $\Hy(z) = \sum_{n\in\ZN}\Hy(\lambda_n)\widehat{\psi}_{n}(z)$. %similar to \eqref{equ:TimeReconstruct}.
Note that it is not possible to simply extract $\Hx$ using \eqref{equ:PlusCos} since $\theta_0$ is unknown.
However we can determine
\begin{align*}
  \Hx'(z)
  &= M \cos(\tfrac{T'}{2}z) - \Hy(z)\,\E^{\I \theta_0}\\
  &=\Hx(z)\, \E^{\I \theta_0} + M(1-e^{\I \theta_0})\cos(\tfrac{T'}{2}z).
\end{align*}
Applying the inverse Fourier transform to $\Hx'$ yields $x'(t) = x(t)e^{\I\theta_0}\, \forall t\in \TT$, using that $T'>T$.
It is therefore possible to reconstruct the time signal $x$ up to a constant phase.
\end{proofsketch}
% ---------------------------------

% ================================================================================================================
% ========== SUMMARY =============================================================================================
% ================================================================================================================
\section{Discussion and Outlook}
\label{sec:Summary}

To determine the sampling system in Fig.\ref{fig:MeasureSetup}, one has to fix $K$, $M$, $a$ and $\beta$.
The number $K \geq 2$ can be chosen arbitrarily. Then $M = K^{2}$ is fixed, and $1\leq a\leq K-1$. %$a$ may have any value between $1$ and $K-1$.
The sampling period $\beta$ has to be chosen such that the sampling system satisfies Condition~B and in particular that $\Lambda^{a}$ is complete interpolating for $\ourPW$.
As discussed before, one possible choice may start with the zeros of the function $\sin(\frac{T'}{2} z)$ with $T'\geq T$.
Then $\delta := \lambda_{k} - \lambda_{k-1} = 2\pi/T'$ such that $\beta = (K-a)\delta$.
Therewith, the total sampling rate becomes
\begin{equation*}
R(a,K) = \frac{M}{\beta} = \frac{K^2}{(K-a)\delta} = \frac{K^2}{K-a} \frac{T'}{2\pi} = \frac{K^2}{K-a}\frac{T'}{T} R_{\mathrm{Ny}}
\end{equation*} 
where $R_{\mathrm{Ny}} := T/(2\pi)$ is the Nyquist rate.
It is apparent that $R(a,K)$ grows asymptotically proportional with $K$ and the oversampling factor $T'/T$ and increases with the overlap $a$. 
The minimal sampling rate $R(a,K) = 4 R_{\mathrm{Ny}}$ is obtained for $K=2$, $T' = T$, and $a=1$.
Since $T'/T$ can be made arbitrarily close to $1$, and in connection with Theorem~\ref{thm:MainThm} and Cor.~\ref{cor:Cor2} this shows that a sampling rate greater $4 R_{\mathrm{Ny}}$ is sufficient for signal recovery based on the magnitude of the signal samples. This corresponds to the findings in \cite{Balan1} for finite dimensional spaces, where it was shown that basically any $x \in \CN^{N}$ can be reconstructed from $M \geq 4N -2$ magnitude samples.

Finally, we note that the above framework can be applied exactly the same way for bandlimited signals. To this end, one only has to exchange the time and frequency domain. Then the modulators in Fig.~\ref{fig:MeasureSetup} have to be replaced by linear filters and the sampling of the magnitudes has to be done in the time domain.

% ================================================================================================================
% ========== BIBLIOGRAPHY ========================================================================================
% ================================================================================================================3
\bibliographystyle{IEEEtran}
\bibliography{IEEEabrv,pub}

\begin{thebibliography}{10}
\providecommand{\url}[1]{#1}
\csname url@samestyle\endcsname
\providecommand{\newblock}{\relax}
\providecommand{\bibinfo}[2]{#2}
\providecommand{\BIBentrySTDinterwordspacing}{\spaceskip=0pt\relax}
\providecommand{\BIBentryALTinterwordstretchfactor}{4}
\providecommand{\BIBentryALTinterwordspacing}{\spaceskip=\fontdimen2\font plus
\BIBentryALTinterwordstretchfactor\fontdimen3\font minus
  \fontdimen4\font\relax}
\providecommand{\BIBforeignlanguage}[2]{{%
\expandafter\ifx\csname l@#1\endcsname\relax
\typeout{** WARNING: IEEEtran.bst: No hyphenation pattern has been}%
\typeout{** loaded for the language `#1'. Using the pattern for}%
\typeout{** the default language instead.}%
\else
\language=\csname l@#1\endcsname
\fi
#2}}
\providecommand{\BIBdecl}{\relax}
\BIBdecl

\bibitem{Oppenheim_Phase_80}
M.~H. Hayes, J.~S. Lim, and A.~V. Oppenheim, ``{Signal Reconstruction from
  Phase or Magnitude},'' \emph{{IEEE} Trans. Acoust., Speech, Signal Process.},
  vol. ASSP-28, no.~6, pp. 672--680, Dec. 1980.

\bibitem{Fienup1982_PhaseRetrieval}
J.~R. Fienup, ``{Phase retrieval algorithms: a comparison},'' \emph{{Applied
  Optics}}, vol.~21, no.~15, pp. 2758--2769, Aug. 1982.

\bibitem{Balan1}
R.~Balan, P.~G. Casazza, and D.~Edidin, ``{On signal reconstruction without
  phase},'' \emph{{Appl. Comput. Harmon. Anal.}}, vol.~20, no.~3, pp. 345--356,
  May 2006.

\bibitem{Balan_Painless}
R.~Balan, B.~G. Bodmann, P.~G. Casazza, and D.~Edidin, ``{Painless
  Reconstruction from Magnitudes of Frame Coefficients},'' \emph{{J. Fourier
  Anal. Appl.}}, vol.~15, no.~4, pp. 488--501, Aug. 2009.

\bibitem{Vetterli2011_PRbeyondFienup}
Y.~M. Lu and M.~Vetterli, ``{Sparse Spectral Factorization: Unicity and
  Reconstruction Algorithms},'' in \emph{{Proc. 36th Intern. Conf. on
  Acoustics, Speech, and Signal Processing (ICASSP)}}, Prague, Czech Republic,
  May 2011, pp. 5976--5979.

\bibitem{CandesEldar_PhaseRetrieval}
E.~J. Cand{\`e}s, Y.~C. Eldar, T.~Strohmer, and V.~Voroninski, ``{Phase
  Retrieval via Matrix Completion},'' \emph{{SIAM Journal on Imaging
  Sciences}}, 2013, to appear.

\bibitem{Thakur2011}
G.~Thakur, ``Reconstruction of bandlimited functions from unsigned samples,''
  \emph{Journal of Fourier Analysis and Applications}, vol.~17, no.~4, pp.
  720--732, 2011.

\bibitem{Young2001_Nonharmonic}
R.~M. Young, \emph{An introduction to nonharmonic Fourier series}.\hskip 1em
  plus 0.5em minus 0.4em\relax Cambridge: Academic Press, 2001.

\bibitem{Christensen_Frames}
O.~Christensen, \emph{{An Introduction to Frames and Riesz Bases}}.\hskip 1em
  plus 0.5em minus 0.4em\relax Boston: Birkh{\"a}user, 2003.

\bibitem{Pavlov_Convergence}
B.~Pavlov, ``Basicity of an exponential system and muckenhoupt's condition,''
  in \emph{Dokl. Akad. Nauk SSSR}, vol. 247, no.~1, 1979, pp. 37--40.

\bibitem{Levin1997_Lectures}
B.~Levin, \emph{Lectures on entire functions}.\hskip 1em plus 0.5em minus
  0.4em\relax American Mathematical Society, 1997.

\bibitem{Levin_Perturbations}
B.~Levin and I.~Ostrovskii, ``Small perturbations of the set of roots of
  sine-type functions,'' \emph{Izv. Akad. Nauk SSSR Ser. Mat}, vol.~43, no.~1,
  pp. 87--110, 1979.

\bibitem{Supper2002_Zeros}
R.~Supper, ``Zeros of entire functions of finite order,'' \emph{Journal of
  Inequalities and Applications}, vol. 2002.

\bibitem{DuffinSchaeffer_38}
R.~Duffin and A.~C. Schaeffer, ``{Some Properties of Functions of Exponential
  Type},'' \emph{{Bull. Amer. Math. Soc.}}, vol.~44, pp. 236--240, Apr. 1938.

\bibitem{MoenichBoche_SP2010}
U.~J. M{\"o}nich and H.~Boche, ``Non-equidistant sampling for bounded
  bandlimited signals,'' \emph{Signal processing}, vol.~90, no.~7, pp.
  2212--2218, 2010.

\end{thebibliography}

\end{document}